\documentclass[10pt]{article}
\usepackage{basic}
\usepackage{typth}
\usepackage{bridge}
\usepackage{verbatim}
\usepackage{pxfonts}
\usepackage{scalerel}

\DeclareMathOperator*\cnc{\scalerel*{\bigcirc}{\bigodot}}
\DeclareMathOperator*\pnc{\bigodot}

\input xy
\xyoption{all}

\title{A proof theoretic study of abstract termination principles}
\author{Thomas Powell}
\date{Preprint, \today}
\begin{document}

\maketitle

\begin{abstract}
We carry out a proof theoretic analysis of the wellfoundedness of recursive path orders in an abstract setting. We outline a very general termination principle and extract from its wellfoundedness proof subrecursive bounds on the size of derivation trees which can be defined in G\"{o}del's system T plus bar recursion. We then carry out a complexity analysis of these terms, and demonstrate how this can be applied to bound the derivational complexity of term rewrite systems.
\end{abstract}

%%%%%%%%%%%%%%%%%%%%%%%%%%%%%%%%%%%%%%%%%%%%%%%%%
%%%%%%%%%%%%%%%%%%%%%%%%%%%%%%%%%%%%%%%%%%%%%%%%%
\section{Introduction}
%%%%%%%%%%%%%%%%%%%%%%%%%%%%%%%%%%%%%%%%%%%%%%%%%
%%%%%%%%%%%%%%%%%%%%%%%%%%%%%%%%%%%%%%%%%%%%%%%%%
\label{sec-intro}

The ability to deduce whether or not a program terminates is crucial in computer science. Though termination is not a decidable property, a number of powerful \emph{proof rules}, or \emph{termination principles}, have been developed, which set out general conditions under which programs can be shown to terminate. Examples of these include path orders for rewrite systems \cite{Dershowitz(1982.0)}, size-change principles \cite{LeeJonBA(2001.0)} and more recently methods based on Ramsey's theorem \cite{PodRyb(2004.0)}.

Any termination principle $P$ gives rise to the following question: Given that a program can be proven to terminate using $P$, can we infer an upper bound on its complexity? This is in turn an instance of a much more general problem captured by Kreisel in his famous quote: \emph{``What more do we know if we have proved a theorem by restricted means than if we merely know the theorem is true?''}. 

In this article I focus on termination via path orders. This area already contains a number of well known complexity results of the above kind. For example, termination via the multiset path ordering implies primitive recursive derivational complexity, while the lexicographic path ordering induces multiple recursive derivation lengths. These bounds were initially established via direct calculations in \cite{Hofbauer(1992.0)} and \cite{Weiermann(1995.0)} respectively. 

Here, I take an approach to complexity closer to the spirit of Kreisel by addressing the following question: Given a proof that some abstract order is wellfounded, can we extract from this proof a subrecursive program which computes derivation sequences and thus provides a bound on their length?

Broadly speaking, there are two ways of accomplishing this. We could choose to concentrate on the \emph{proof}, showing that it can be formalised in some weak theory and then appealing to an appropriate logical metatheorem which would guarantee that a realizing program in some corresponding system of functionals exists. Alternatively, we could directly extract a realizing term by hand and show that it can be defined in some subrecursive calculus. The latter approach is chosen here: Exhibiting an explicit realizing term is not only more illuminating, but we can appeal to mathematical properties of that term to obtain more refined complexity results.

The starting point of this work is the elegant paper of Buchholz \cite{Buchholz(1995.0)}, who was the first to apply proof theory in the style of Kreisel to termination principles. More specifically, he rederived the aforementioned bounds on the multiset and lexicographic orders by showing that wellfoundedness of these orders could be formalised in weak fragments of Peano arithmetic, and then applying a program extraction theorem to obtain the corresponding bounds on the length of reduction sequences. Key to Buchholz's method is to consider finitely branching variants of the usual path orders - an approach which will be essential to us as well.

My second source of inspiration is the recent collection of papers (including \cite{BerOliSte(2015.0),BerSte(2015.0),FriYokSte(2017.0)}), which study both size-change termination and techniques based on Ramsey's theorem from the perspective of proof theory. In particular, in \cite{BerOliSte(2015.0)}, an upper bound on the length of transition sequences is given as a term of System T extended with bar recursion. It turns out that bar recursion - that is recursion over wellfounded trees - is a form of recursion naturally well suited to computing derivation trees of programs. Moreover, where complexity is concerned, one can directly appeal to \emph{closure properties} of bar recursion (see the recent \cite{OliStei(2017.0)}) to establish upper bounds on the size of these trees.

In this paper, I study an abstract termination principle which subsumes the majority of path orders encountered in the literature. It is closely related to the first termination theorem considered by Goubaul-Larrecq in \cite{GLar(2001.0)}, though here it is based on orderings which are assumed to be finitely branching. We give a classical proof of the theorem, which we then analyse to show that given a moduli which forms a computational analogue of the theorem's main condition, a function bounding the size of derivations can be defined using bar recursion of lowest type. A number of initial complexity results can already be given by appealing to \cite{OliStei(2017.0)} and related works.

I then consider a variant of the theorem in which a computationally stronger realizer to the premise is given. In this case, more refined complexity results are possible, which are set out in Corollary \ref{prclosure}. I conclude by showing how the well known upper bounds for the complexity of simplification orders follow from this result, and illustrating how Buchholz's finitely branching orders fit in to our general computational framework.

My hope is that the results of this paper form a framework for complexity which will can be developed further in the future, with potential for both more general and more refined results. In addition, in the process of our proof theoretic analysis we explore a number of deep mathematical concepts which underlie path orders, including minimal-bad-sequence style constructions, realizability and bar recursion, connections between the latter having been explored from a more general perspective in e.g. \cite{Powell(2018.0),Seisenberger(2003.0)}. I aim to demonstrate how these concepts all come together to form a particularly elegant illustration of the bridge between proofs and programs.

%%%%%%%%%%%%%%%%%%%%%%%%%%%%%%%%%%%%%%%%%%%%%%%%%
\subsection{Related work}
%%%%%%%%%%%%%%%%%%%%%%%%%%%%%%%%%%%%%%%%%%%%%%%%%
\label{sec-intro-related}

This article forms a considerable generalisation of the results presented in \cite{MosPow(2015.0)}, which is concerned exclusively with the multiset and lexicographic path orders. An abstract termination principle is also studied in \cite{Powell(2019.0)}, but there no restrictions to the logical complexity of the order are made, and the focus is on finding equivalent formulations of the axiom of dependent choice in all finite types.

%%%%%%%%%%%%%%%%%%%%%%%%%%%%%%%%%%%%%%%%%%%%%%%%%
\subsection{Prerequisites and notation}
%%%%%%%%%%%%%%%%%%%%%%%%%%%%%%%%%%%%%%%%%%%%%%%%%
\label{sec-intro-prereq}

We assume that the reader is familiar with G\"{o}del's System T of primitive recursive functionals in all finite type, which we will use as our base programming language. We will actually use a fairly rich formulation of System T, which includes both product type $\rho\times\tau$ together with finite sequence types $\rho^\ast$. We denote by $\sT_i$ the fragment of System T which only permits recursion of type level $i$. We will use the following notation:
\begin{itemize}

\item We denote by $0_\rho$ the canonical zero element of type $\rho$, defined in the obvious way (we set $0_{\rho^\ast}=[]$ for sequence types).

\item $|a|$ is the length of the sequence $a\in\rho^\ast$;

\item if $a=[a_0,\ldots,a_{k-1}]$ then $a\ast x:=[a_0,\ldots,a_{k-1},x]$ denotes the concatenation of $a$ with $x$, similarly $x\ast a:=[x,a_0,\ldots,a_{k-1}]$;

\item $\bar{a}:=a_{k-1}$ denotes the last element of $a$ (we just set $\bar{a}=0_\rho$ if $a=[]$);

\item we write $x\in a$ if $x=a_i$ for some $i<|a|$;

\item for $\alpha:\NN\to\rho$ we have $\initSeg{\alpha}{n}:=[\alpha_0,\ldots,\alpha_{n-1}]$;

\item for $a\in \rho^\ast$ we define $\ext{a}\in \NN\to\rho$ by $\ext{a}_n:=a_n$ if $n<|a|$ and $\ext{a}_n:=0_X$ otherwise.

\end{itemize}
At several points we will need to extend $\sT$ with constants $\rec^{\lhd,\rho}$ for wellfounded recursion of output type $\rho$ over some decidable wellfounded relation $\lhd$ on $\NN$, which will satisfy the defining axiom
\begin{equation*}
\rec^{\lhd,\rho}_f(x)=fx(\lambda y\lhd x\; . \; \rec^{\lhd,\rho}_f(y))
\end{equation*}
where $\lambda y\lhd x\; . \; g(y)$ is shorthand for `if $y\lhd x$ then $g(y)$ else $0$'. When defining recursive functionals we typically use the convention, as above, of writing parameters which don't change in the defining equation as a subscript.

%%%%%%%%%%%%%%%%%%%%%%%%%%%%%%%%%%%%%%%%%%%%%%%%%
%%%%%%%%%%%%%%%%%%%%%%%%%%%%%%%%%%%%%%%%%%%%%%%%%
\section{Finitely branching relations}
%%%%%%%%%%%%%%%%%%%%%%%%%%%%%%%%%%%%%%%%%%%%%%%%%
%%%%%%%%%%%%%%%%%%%%%%%%%%%%%%%%%%%%%%%%%%%%%%%%%
\label{sec-finitary}

We start of in this section by introducing some basic facts and definitions concerning finitely branching relations in general, and introduce the concept of bar recursion.

Our basic object of study will be a binary relation $\succ$ on some set $X$. In the context of termination analysis, $X$ is typically a set of terms in some language. A program $P$ is then considered to be reducing with respect to $\succ$ if whenever $$t_0\leadsto_P t_1\leadsto_P \ldots\leadsto_P t_k$$ is a run on $P$, then $t_i\succ t_{i+1}$ for all $i<k$. Thus wellfoundedness of $\succ$ implies that the program terminates.

However, up until Section \ref{sec-trs}, everything will be carried out in an abstract setting. For now, the only assumption we make about $X$ is that it can be \emph{arithmetized} i.e. comes equipped with some bijective encoding $\goedel{\cdot}:X\to \NN$, and similarly $\succ$ is a primitive recursive relation i.e. there is some term $r:\NN\times\NN\to\NN$ definable in $\sT_0$ such that $x\succ y$ iff $r(\goedel{x},\goedel{y})=0$. For the sake of clarity, we continue to refer to the set as $X$ rather than $\NN$, but it should be remembered that for practical purposes $\succ$ is a relation on natural numbers, and this will indeed be crucial when we come to our complexity results later.

In this paper we will primarily be concerned with relations which are finitely branching.
\begin{definition}
\label{def-fb}
We say that $\succ$ is \emph{finitely branching} if 
\begin{equation*}
(\forall x\in X)(\exists a\in X^\ast)(\forall y)(x\succ y\leftrightarrow y\in a).
\end{equation*}
In particular, the number of distinct elements $y$ with $x\succ y$ is bounded above by $|a|$.
\end{definition}
We now need to give a precise definition of what we mean by \emph{wellfoundedness}. We will primarily be interested in the following formulation, which in \cite{BerOliSte(2015.0)} is referred to as \emph{classical wellfoundedness}:
\begin{definition}
We call a sequence $\alpha\in X^\NN$ (classically) wellfounded and write $W(\alpha)$ if
\begin{equation*}
\exists n(\alpha_n\nsucc \alpha_{n+1})
\end{equation*}
where $x\nsucc y$ denotes $\neg(x\succ y)$.
Similarly, we say that an element $x\in X$ is wellfounded, and also write $W(x)$, if
\begin{equation*}
\forall\alpha(x=\alpha_0\to W(\alpha)).
\end{equation*}
The relation $\succ$ is wellfounded if $(\forall x) W(x)$.
\end{definition}
In the next Section we will consider an equivalent formulation of wellfoundedness which is classically equivalent to the above but computationally stronger. We now make precise what we mean by the \emph{complexity} of some object in $X$.
\begin{definition}
[Finite chain]
We call a finite sequence $a\in X^\ast$ is a $\succ$-chain, and write $C_\succ(a)$, if $a_i\succ a_{i+1}$ for all $i<|a|-1$.
\end{definition}
\begin{definition}
[Derivational complexity]
\label{defn-dc}
Let $x\in X$ and suppose that there exists some $k$ such that 
\begin{equation*}
C_\succ(x\ast a)\to |a|\leq k.
\end{equation*}
We call the minimal such $k$ the \emph{derivational complexity} of $x$ and denote it by $\dc(x)$. We say that the derivational complexity of some wellfounded $\succ$ is bounded by some function $f:X\to \NN$ if $\dc(x)\leq f(x)$ for all $x\in X$.
\end{definition}
We now give a syntactic formulation of wellfoundedness which will be crucial to us later, and which is adapted from Buchholz's notion of a derivation \cite{Buchholz(1995.0)}. Here we work with a structure which encodes in a slightly more precise way the derivation tree generated by some wellfounded $x$.  
\begin{definition}
[Derivation tree]
\label{defn-dt}
The predicate $T(x,d)$ on $X\times X^\ast$ is defined by induction on the length of $d$ as follows: If $[y_0,\ldots,y_{k-1}]$ is the unique sequence consisting exactly of those elements $y$ with $x\succ y$, with $y_i<y_j$ iff $i<j$, then $T(x,d)$ holds precisely when $d=x\ast d_0\ast\ldots\ast d_{k-1}$ and $T(y_i,d_i)$ holds for all $i<k$.
\end{definition}
Intuitively, $T(x,d)$ holds iff $d$ represents the flattening of the derivation tree of $x$ which would be obtained by a depth first search and ordering each child node by its encoding. Take for example the order on $\{1,2,\ldots,7\}$ defined by
\begin{equation}
\label{eqn-ex}
2\succ 4,7 \ \ \ \ \ \ 4\succ 1,3,6 \ \ \ \ \ \ 3\succ 5
\end{equation}
Then we would have $T(2,d)$ iff $d=[2,4,1,3,5,6,7]$. Note that $T(x,d)$ makes sense even when $x$ is not wellfounded: in that case $T(x,d)$ would simply be false for all $d$. However, when it holds for some $d$ then this must be unique.
\begin{lemma}
\label{lem-dt!}
If $T(x,d)$ and $T(x,e)$ then $d=e$.
\end{lemma}
\begin{lemma}
\label{lem-dt-W}
If $T(x,d)$ then $\dc(x)\leq |d|$.
\end{lemma}
\begin{proof}
Induction on the length of $d$. Suppose that $x\succ x_1\succ\ldots\succ x_k$. Then we have $T(x_1,e)$ for some $e$ contained in $d$, and assuming inductively that $k-1\leq \dc(x_1)\leq |e|<|d|$ we obtain $k\leq |d|$ and thus $\dc(x)\leq |d|$. 
\end{proof}
\begin{theorem}
\label{thm-W-dt}
If $\succ$ is finitely branching then $W(x)$ holds iff $(\exists d) T(x,d)$.
\end{theorem}

\begin{proof}
One direction follows immediately from Lemma \ref{lem-dt-W}. For the other, $(\forall d)\neg T(x,d)$ would imply that the derivation tree of $x$ is infinite, and so by K\"{o}nig's lemma this tree must have an infinite branch. But that would contradict $W(x)$.
\end{proof}

%%%%%%%%%%%%%%%%%%%%%%%%%%%%%%%%%%%%%%%%%%%%%%%%%
\subsection{Computing derivation trees}
%%%%%%%%%%%%%%%%%%%%%%%%%%%%%%%%%%%%%%%%%%%%%%%%%
\label{sec-finitary-comp}

The focus of this article will be on the construction of explicit functions $\Phi:X\to X^\ast$ for wellfounded relations $\succ$ such that $(\forall x)T(x,\Phi(x))$. Such a $\Phi$ will be called a \emph{derivation function} for $\succ$. Whenever $\Phi$ is a derivation function for $\succ$, by Lemma \ref{lem-dt-W} in particular it follows that the map $\lambda x.|\Phi(x)|$ bounds the derivational complexity of $\succ$. Therefore, whenever we can guarantee that $\Phi$ can be defined in some restricted class of functions, we can produce a subrecursive bound the derivational complexity of $\succ$.

Note that for any finitely branching $\succ$, provided we know in advance that $x$ is wellfounded, its derivation tree can be easily computed by simply implementing a depth first search. However, a much stronger result would be to show that the computation of a derivation tree can be defined in some \emph{subrecursive} calculus, which takes into account the strength of the system in which $W(x)$ can be proved.

In this section we give a short and simple result of this kind, where we analyse the statement
\begin{equation}
\label{eqn-KL}
\mbox{if $\succ$ is finitely branching and wellfounded then $(\forall x)(\exists d) T(x,d)$}
\end{equation}
The statement follows as in Theorem \ref{thm-W-dt} from an application of K\"{o}nig's lemma. We are interested in giving (\ref{eqn-KL}) a \emph{computational interpretation}, namely a construction of a derivation function for $\succ$ which takes as parameters some functionals which give a computational interpretation to the premise of (\ref{eqn-KL}). This leads us to the following key definitions:
\begin{definition}
\label{defn-moduli}
\begin{enumerate}[(a)]

\item A \emph{branching modulus} for $\succ$ is a function $c:X\to X^\ast$ satisfying
\begin{equation*}
x\succ y\leftrightarrow y\in c(x)
\end{equation*}
for all $x,y\in X$. We assume w.l.o.g. that $c(x)$ is ordered and contains no repetitions.

\item A \emph{modulus of wellfoundedness} for $\succ$ is a function $\omega:X^\NN\to\NN$ satisfying
\begin{equation*}
(\exists i<\omega(\alpha))(\alpha_i\nsucc \alpha_{i+1})
\end{equation*}
for all $\alpha\in X^\NN$.
\end{enumerate}
\end{definition}
Moduli of wellfoundedness are also been studied in \cite{BerOliSte(2015.0)} in the context of the Podelski-Rybalchenko termination theorem, and we take our terminology from them.

Given a branching modulus and $x\in X$, one can easily compute the derivation tree $d$ for $x$ by implementing a depth first search. We now show that given, in addition, a modulus of wellfoundedness, we can give a subrecursive definition of the derivation function in $\sT$ plus bar recursion, where the latter is a recursion scheme over wellfounded trees. There are numerous different variants of bar recursion (see \cite{Powell(2014.0)}), but here we will be primarily concerned with the original version due to Spector \cite{Spector(1962.0)}.
\begin{definition}
[Bar recursion]
The constant $B^{\rho,\tau}$ of bar recursion of type $\rho,\tau$ is characterised by the following defining equation (cf. Section \ref{sec-intro-prereq} for notation):
\begin{equation*}
B^{\rho,\tau}_{\omega,g,h}(a):=\begin{cases} g(a) & \mbox{if $\omega(\ext{a})<|a|$}\\ h a(\lambda x. B^{\rho,\tau}_{\omega,g,h}(a\ast x)) & \mbox{otherwise}\end{cases}
\end{equation*}
where $a\in \rho^\ast$ and the other parameters have types $\omega:\rho^\NN\to\NN$, $g:\rho^\ast\to \tau$ and $h:\rho^\ast\to (\rho\to \tau)\to \tau$.
\end{definition}
Before we give our construction, we need some notation for some simple recursive operations on sequences. We use the symbol $\cnc$ to denote iterated list concatenation. So for $a=[a_0,\ldots,a_{k-1}]$ we have 
\begin{equation*}
\cnc_{x\in a} x:=a_0\ast\ldots\ast a_{k-1}.
\end{equation*}
We abuse this symbol just like a summation symbol, so for example
\begin{equation*}
\cnc_{x\in a} p(x):=p(a_0)\ast\ldots\ast p(a_{k-1})
\end{equation*}
and so on. Note that this operation is definable in $\sT_0$ using recursion over the length of $|a|$.

\begin{lemma}
\label{lem-treecnc}
Let $x\in X$ and suppose that $p:X\to X^\ast$ is a function satisfying $T(y,p(y))$ for all $y\prec x$. Then
\begin{equation*}
d:=x\ast\cnc_{y\in c(x)} p(y)
\end{equation*}
satisfies $T(x,d)$ whenever $c$ is a branching modulus for $\succ$.
\end{lemma}

\begin{proof}
Directly from the Definitions \ref{defn-dt} and \ref{defn-moduli} (a).
\end{proof}

\begin{theorem}
\label{thm-tree}
In $\sT_0+\BR^{X,X^\ast}$ we can define a function $\Psi_{c,\omega} X^\ast\to X^\ast$ which takes parameters $c:X\to X^\ast$ and $\omega:X^\NN\to\NN$ and satisfies 
\begin{equation*}
\Psi_{c,\omega}(a):=\begin{cases}[] & \mbox{if $|a|=0$ or $\omega(\ext{a})<|a|$}\\ \bar{a}\ast\cnc_{y\in c(\bar{a})}\Psi_{c,\omega}(a\ast y) & \mbox{otherwise}.\end{cases}
\end{equation*}
Moreover, if $c$ resp. $\omega$ is a branching modulus resp. modulus of wellfoundedness for $\succ$, then the function $\lambda x.\Psi_{c,\omega}([x])$ is a derivation function for $\succ$. 
\end{theorem}

\begin{proof}
That $\Psi$ is definable in $\sT_0+\BR^{X,X^\ast}$ is a simple exercise, and we omit it here (though definability results in later sections are included in full in the appendix). For the verification proof, we first show that for any sequence satisfying $|a|>0$ and $C_\succ(a)$ we have:
\begin{equation}
\label{eqn-tree}
\neg T(\bar{a},\Psi_{c,\omega}(a))\to (\exists y\prec a) \neg T(y,\Psi_{c,\omega}(a\ast y))
\end{equation}
To see this, note that $C_\succ(a)$ implies that $\omega(\ext{a})\geq |a|$, else there would be some $i,i+1<|a|$ with $a_i\nsucc a_{i+1}$. Therefore by the contrapositive of Lemma \ref{lem-treecnc} we obtain (\ref{eqn-tree}).

Now, suppose that there exists some $x$ such that $\neg T(x,\Psi_{c,\omega}([x]))$. Then by dependent choice together with (\ref{eqn-tree}) there exists some infinite descending sequence $\alpha_0\succ\alpha_1\succ\ldots$, contradicting wellfoundedness of $\succ$. Thus $T(x,\Psi_{c,\omega}([x]))$ holds for all $x$, and we're done.
\end{proof}

The above theorem is not deep it itself, but we included as a simple illustration of the results which will follow. Note the verification proof uses classical logic together with dependent choice: We could convert this into an intuitionistic proof which instead uses some variant of bar induction, as is typically the case for program extraction theorems. However, in this paper we have no deeper foundational goals which would require the verification of our extracted terms to be formalisable in a weak intuitionistic theory, so we stick to classical logic as it is usually more intuitive.

%%%%%%%%%%%%%%%%%%%%%%%%%%%%%%%%%%%%%%%%%%%%%%%%%
%%%%%%%%%%%%%%%%%%%%%%%%%%%%%%%%%%%%%%%%%%%%%%%%%
\section{Abstract path orders}
%%%%%%%%%%%%%%%%%%%%%%%%%%%%%%%%%%%%%%%%%%%%%%%%%
%%%%%%%%%%%%%%%%%%%%%%%%%%%%%%%%%%%%%%%%%%%%%%%%%
\label{sec-simp}

Path orders form one of the earliest proof rules for termination, and are a central concept in the theory of term rewriting. Today, a huge variety of different path orders have been developed, ranging from the general - such as the unified ordering of \cite{YamKS(2013.0)} - which focus on the common structure shared by termination orders, to the highly specialised - such as the polynomial path ordering of \cite{AvanMos(2013.0)} - which aim to capture a very precise class of terminating programs. We talk about path orders in more detail in Section \ref{sec-trs}, but for now we give a simple explanation which helps motivate the abstract principle studied here.

Very roughly, path orders capture `termination via minimal sequences'. Consider the Ackermann-P\'{e}ter function, which recursively defined by the rules
\begin{equation*}
\begin{aligned}
A(0,n)&\succ n+1 \\
A(m,0)&\succ A(m-1,1)\\
A(m+1,n+1)&\succ A(m,A(m+1,n))
\end{aligned}
\end{equation*}
Suppose that $A(m,n)$ is not wellfounded i.e. triggers an infinite computation. Then either $A(m,n-1)$ is not wellfounded, or there is some $k$ such that $A(m-1,k)$ is not wellfounded. In other words, there is some $(m',n')$ lexicographically less than $(m,n)$ such that $A(m',n')$ is not wellfounded. Thus non-wellfoundedness of $A(m,n)$ gives rise to an infinite sequence
\begin{equation}
\label{eqn-minimal}
A(m_0,n_0)\gg A(m_1,n_2)\gg A(m_2,n_2)\gg\ldots
\end{equation}
where $A(m,n)\gg A(m',n')$ denotes that $(m',n')$ is lexicographically smaller than $(m,n)$. Informally speaking, (\ref{eqn-minimal}) plays the role of a minimal sequence, in the sense that it represents instances of $A$ whose arguments i.e. subterms are wellfounded, but which are themselves non-wellfounded. Since the relation $\gg$ is wellfounded, then we have proven totality of the Ackermann function.

On an abstract level, path orders are a proof rule which implement the idea that termination of a program can be inferred from wellfoundedness of minimal sequences.
\begin{definition}
Let $\rhd$ be a primitive recursive relation on $X$ which is \emph{inductively wellfounded}, by which we mean that induction and recursion over $\rhd$ is available.
\end{definition}
Inductive wellfoundedness is equivalent to classical wellfoundedness as defined in Section \ref{sec-finitary}. However, from a computational point of view the two differ: Classical wellfoundedness is realized by some modulus of wellfoundedness of type $X^\NN\to\NN$, while the computational analogue of inductive wellfoundedness will be the recursor $\rec^\rhd$. Note that a modulus of wellfoundedness for $\rhd$ is easily computable in $\rec^{\rhd}$, but defining $\rec^\rhd$ in some modulus of wellfoundedness for $\rhd$ would seem to require bar recursion in addition.

The reason that we choose to $\rhd$ to be inductively wellfounded is that when $X$ is some set of terms, $\rhd$ usually represents the subterm relation, recursion over which is trivially definable in $\sT_0$. A key concept in our abstract termination principle is the notion of a \emph{minimal sequence}. The precise definition is as follows:

\begin{definition}
[Minimal sequence]
An infinite sequence $\alpha\in X^\NN$ is minimal if $W(y)$ for all $y\lhd \alpha_n$ and $n\in\NN$.
\end{definition}

In addition to $\succ$ and $\rhd$ we consider some auxiliary order $\gg$, which interacts with the other relations in a specific way, which we call a \emph{decomposition} after the similar notion in \cite{FerZan(1995.0)}.

\begin{definition}
[Decomposition]
\label{defn-decomp}
A primitive recursive relation $\gg$ on $X$ is called a decomposition of $\succ$ w.r.t $\rhd$ if it satisfies
\begin{enumerate}[(i)]

\item\label{decompi} whenever $x\succ y$ then either $x\gg y$ or there exists some $z\lhd x$ such that $z\succeq y$;

\item\label{decompii} whenever $x\gg y$ and $y\rhd z$ then $x\succ z$.

\end{enumerate}
\end{definition}
We are now ready to state and prove our main abstract termination principle, which is closely related to Theorem 1 of \cite{GLar(2001.0)}.
\begin{theorem}
[Abstract termination principle]
\label{thm-abs}
Suppose that $\gg$ is a decomposition of $\succ$ w.r.t. $\rhd$ which is classically wellfounded on the set of all minimal sequences. Then $\succ$ is wellfounded on $X$.
\end{theorem}

\begin{proof}
Defining $A:=\{x\in X\; : \; (\forall y\lhd x) W(y)\}$, we claim that for any nonempty sequence $a\in A^\ast$ satisfying $C_\gg(a)$ we have:
\begin{equation}
\label{eqn-abs}
\neg W(\bar{a})\to (\exists y\ll \bar{a})(\neg W(y)\wedge y\in A).
\end{equation}
To see this, observe that $\neg W(\bar{a})$ implies that the set
\begin{equation*}
S_{\bar{a}}:=\{x\in X\; | \; x\prec a\wedge \neg W(x)\}
\end{equation*}
is nonempty. Thus by the minimum principle on $\rhd$, which follows classically using induction on $\rhd$, $S_{\bar{a}}$ has some minimal element $y$. Now, it follows that $\bar{a}\gg y$, otherwise, by decomposition property (\ref{decompi}) we would have $\bar{a}\rhd z\succeq y$ for some $z$, and since $\bar{a}\in A$ this contradicts $\neg W(y)$. But using property (\ref{decompii}) we can therefore also show that $y\in A$: since $\bar{a}\gg y$ then for any $z\lhd y$ we have $\bar{a}\succ z$, and $\neg W(z)$ would imply that $z\in S_{\bar{a}}$, contradicting minimality of $y$. This proves the claim.

For the main result, suppose that $\neg W(x)$ holds for some $x$, and define $\alpha_0$ to be the minimal such $x$ w.r.t $\rhd$. Then we have $\alpha_0\in A$, and $C_\gg([\alpha_0])$ trivially, and by applying dependent choice together with (\ref{eqn-abs}) we obtain an infinite sequence $\alpha_0\gg \alpha_1\gg\alpha_2\gg\ldots$ with $\alpha_i\in A$ for all $i$. But the assumption that $\gg$ is wellfounded on minimal sequences.
\end{proof}

We now give a computational interpretation of Theorem \ref{thm-abs}, in the case where both $\succ$ and $\rhd$ are finitely branching. Similarly to before this assumption will be represented by a pair of branching moduli $c_\succ$ and $c_\rhd$. The computational analogue of inductive wellfoundedness of $\rhd$ will be access to wellfounded recursion over $\rhd$, so it remains to formulate our main assumption that $\gg$ is classically wellfounded on the set of all minimal sequences.

\begin{definition}
The predicate $T_\rhd(x,u)$ on $X\times\arr{X}$ is defined as follows:
\begin{equation*}
T_\rhd(x,u):\equiv |u|=k\wedge (\forall i<k) T(y_i,u_i)
\end{equation*}
where $[y_0,\ldots,y_{k-1}]=c_\rhd(x)$.
\end{definition}

Continuing with our earlier example (\ref{eqn-ex}), suppose that $x\rhd y$ only when $y$ is a proper divisor of $x$. Then $T_\rhd(6,[[1],[3,5]])$, since $1$ and $3$ are the only proper subdivisors of $6$ and both $T(1,[1])$ and $T(3,[3,5])$. 

\begin{lemma}
A sequence $\alpha\in X^\NN$ is minimal iff there exists a sequence $\beta\in (\arr{X})^\NN$ such that $T_\rhd(\alpha_n,\beta_n)$ holds for all $n\in\NN$.
\end{lemma}

\begin{proof}
Directly from Theorem \ref{thm-W-dt}.
\end{proof}

This syntactic characterisation for finitely branching orders informs the following adaptation of the modulus of wellfoundedness:
\begin{definition}
\label{defn-minwf}
A modulus of minimal wellfoundedness for $\succ$, $\rhd$ and $\gg$ is a function $\omega:(X\times \arr{X})^\NN\to\NN$ satisfying
\begin{equation*}
(\forall n) T_\rhd(\alpha_n,\beta_n)\to (\exists i<\omega(\alpha,\beta))(\alpha_i\ngg\alpha_{i+1})
\end{equation*}
where by for clarity we represent the two components of $(X\times \arr{X})^\NN$ separately as $\alpha\in X^\NN$ and $\beta\in (\arr{X})^\NN$, and write e.g. $\omega(\alpha,\beta)$ instead of $\omega(\lambda i.\pair{\alpha_i,\beta_i})$.
\end{definition}
In the construction that follows we denote by $\pnc$ the usual map function i.e. given $a\in X^\ast$ and $p:X\to Y$ we have
\begin{equation*}
\pnc_{x\in a}p(x):=[p(a_0),\ldots,p(a_{k-1})]\in Y^\ast
\end{equation*}
where $a=[a_0,\ldots,a_{k-1}]$. In addition, given two lists $a,a'\in X^\ast$ we denote by $a\cap a'$ the ordered intersection of these lists. The following lemma follows directly from the definitions:

\begin{lemma}
\label{lem-map}
Let $x\in X$ and suppose that $q:X\to X^\ast$ is a function satisfying $T(y,q(y))$ for all $y\lhd x$. Then
\begin{equation*}
u:=\pnc_{y\in c_\rhd(x)}q(y)
\end{equation*}
satisfies $T_\rhd(x,u)$ whenever $c_\rhd$ is a branching modulus for $\rhd$.
\end{lemma}

\begin{lemma}
\label{lem-absdef}
Let $Y:=X\times\arr{X}$ and $c_\succ,c_\rhd:X\to X^\ast$ are some fixed terms of $T_0$ which form branching moduli for $\succ$ and $\rhd$. Then there is a functional $\Psi:(Y^\NN\to\NN)\to Y^\ast\to X^\ast$ definable in $T_0+\rec^{\rhd,X^\ast}+\BR^{Y,X^\ast}$ which satisfies
\begin{equation*}
\Psi_\omega(a,b)=\begin{cases}[] & \mbox{if $|a|=0$ or $\omega(\widehat{a,b})<|a|$}\\ \bar{a}\ast\cnc_{y\in c_\succ(\bar{a})} R_{a,b}(y) & \mbox{otherwise} \end{cases}
\end{equation*}
where $R_{a,b}:X\to X^\ast$ in turn satisfies
\begin{equation*}
R_{a,b}(y):=\begin{cases}\bar{b}_i[y] & \mbox{if $y\preceq c_\rhd(\bar{a})_i$ for some $i<|c_\rhd(\bar{a})|$}\\ \Psi_{\omega}(a\ast y,b\ast\pnc_{z\in c_\rhd(y)} R_{a,b}(z)) & \mbox{otherwise}\end{cases}
\end{equation*}
where in the first line, for $d\in X^\ast$ and $y\in X$, $d[y]\subset d$ denotes some sequence contained in $d$ and satisfying $T(y,d[y])$ whenever it exists (and just $[]$ otherwise).
\end{lemma}

\begin{proof}
Routine: See Appendix \ref{sec-app} for full details.
\end{proof}

For the purposes of our verification proof, we make a small assumption: That $0$ encodes some object of $X$ which is minimal w.r.t. $\lhd$ i.e. contains no subterms. In particular, we would have $T_\rhd(0,0_{X^\ast})$ since $0_{X^\ast}$ is assumed to be the empty sequence. While not strictly necessary, this assumption allows us to use the usual variant of bar recursion as above. 

\begin{theorem}
\label{thm-abscomp}
Suppose that $\Psi$ is defined as in Lemma \ref{lem-absdef}, and define $\Phi:(Y^\NN\to\NN)\to X\to X^\ast$ from $\Psi$ over $\sT_0+\rec^{\rhd,X^\ast}$ as
\begin{equation*}
\Phi_\omega(x)=\Psi_\omega([x],[\pnc_{y\in c_\rhd(x)}\Phi_\omega(y)]).
\end{equation*}
Then whenever $\omega$ is a modulus of minimal wellfoundedness for $\succ$, $\rhd$ and $\gg$ then $\Phi_\omega$ is a derivation function for $\succ$.
\end{theorem}

\begin{proof}
We first claim that for any nonempty $\pair{a,b}\in Y^\ast$ such that $T_\rhd(a_i,b_i)$ for all $i<|a|$ and $C_\gg(a)$ then
\begin{equation}
\label{eqn-abscomp}
\neg T(\bar{a},\Psi_{\omega}(a,b))\to (\exists y\ll \bar{a},u)(\neg T(y,\Psi_\omega(a\ast y,b\ast u))\wedge T_\rhd(y,u)).
\end{equation}
To prove the claim, we begin by observing that $\omega(\widehat{a,b})\geq |a|$. To see this, observe that $\hat{a}$ is a minimal sequence relative to $\hat{b}$, by our assumption that $T_\rhd(0,[])$ holds. Thus $\omega(\widehat{a,b})<|a|$ would imply that there exist $i,i+1<|a|$ such that $a_i\gg a_{i+1}$, contradicting $C_\gg(a)$.

Therefore $\Psi_\omega(a,b)=\bar{a}\ast\cnc_{y\in c_\rhd(\bar{a})} R_{a,b}(y)$ and by Lemma \ref{lem-treecnc} there exists some $y\prec \bar{a}$ such that $\neg T(y,R_{a,b}(y))$ and so the set
\begin{equation*}
S_{a,b}:=\{x\in X\; : \; x\prec\bar{a}\wedge \neg T(x,R_{a,b}(x))\}
\end{equation*}
is nonempty. By the minimum principle $S_{a,b}$ contains some minimal $y$. Let $[z_0,\ldots,z_{k-1}]:=c_\rhd(\bar{a})$. If $u\preceq z_i\lhd \bar{a}$ for some $i<k$, then since $T_\rhd(\bar{a},\bar{b})$ and thus $T(z_i,\bar{b}_i)$ we would have $T(y,\bar{b}_i[y])$ and thus $T(y,R_{a,b}(y))$, a contradiction. Therefore as before $y\ll \bar{a}$ by decomposition property (\ref{decompi}). Now for $z\in c_\rhd(y)$, by property (\ref{decompii}) we have $z\prec\bar{a}$ and thus $T(z,R_{a,b}(z))$ by minimality of $y$. Therefore by Lemma \ref{lem-map}, $u:=\pnc_{z\in c_\rhd(y)} R_{a,b}(z)$ satisfies $T_\rhd(y,u)$ and from $\neg T(y,R_{a,b}(y))$ we obtain $\neg T(y,\Psi_\omega(a\ast y,b\ast u))$. This proves the claim.

Now suppose the theorem is false and take some minimal $x$ such that $\neg T(x,\Phi_\omega(x))$. Then $T_\rhd(x,v)$ and $\neg T(x,\Psi_\omega([x],[v]))$ hold for $v:=\pnc_{y\in c_\rhd{x}}\Phi_\omega(y)$, and by dependent choice in conjunction with (\ref{eqn-abscomp}) we obtain a pair of sequences $\alpha,\beta$ such that $T_\rhd(\alpha_n,\beta_n)$ but $\alpha_n\gg \alpha_{n+1}$ for all $n$, a contradiction.
\end{proof}

%%%%%%%%%%%%%%%%%%%%%%%%%%%%%%%%%%%%%%%%%%%%%%%%%
\subsection{Primitive recursive bounds via closure results for bar recursion}
%%%%%%%%%%%%%%%%%%%%%%%%%%%%%%%%%%%%%%%%%%%%%%%%%
\label{sec-simp-barclosure}

Having extracted a bar recursive term which computes derivation trees for $\succ$, we can already apply a variety of closure results in the literature to obtain crude upper bounds on the derivational complexity of $\succ$. The term $\Psi_{\omega}$ in Lemma \ref{lem-absdef} is formally definable not just from bar recursion but from a \emph{single instance} of $\BR^{X\times \arr{X},X^\ast}_{\omega,g,h}$ where $g$ and $h$ are definable in $\sT_0+\wrec^{\rhd,X^\ast}$.  

As a consequence, we can show that the derivational complexity of $\succ$ is bounded by some G\"{o}del primitive recursive function whenever the modulus of minimal wellfoundedness is definable in System T. This follows directly from Schwichtenberg's classic result \cite{Schwichtenberg(1979.0)} that System T is closed under the rule of bar recursion, whenever bar recursion has sequence type level $0$ or $1$. A more fine grained analysis is the following:
\begin{corollary}
\label{barclosure}
Suppose that $\succ$ is a binary relation whose branching modulus for both $\succ$ and $\rhd$ is definable in $\sT_0$, and that $\rec^{\rhd,X^\ast}$ is also definable in $\sT_0$.
\begin{enumerate}[(a)]

\item\label{barclosa} Whenever $\gg$ has a modulus of minimal wellfoundedness $\omega$ which is definable in $\sT_i$, the derivational complexity of $\succ$ is bounded by some function in $\sT_{i+3}$.

\item\label{barclosb} In the special case where $\omega$ is definable in $\sT_0$, the derivational complexity is bounded by some function in $\sT_1$.

\end{enumerate}

\end{corollary}

\begin{proof}
Since $X$ is coded in the natural numbers, both the sequence type $X\times\arr{X}$ and the output type $X^\ast$ can also be encoded in $\NN$, and so the functional $\Psi_{\omega}$ is definable from a single instance $\BR_{\omega,g,h}$ of bar recursion of lowest type. By the recent analysis of Oliva and Steila \cite{OliStei(2017.0)}, whenever the parameters $g,h$ are in $\sT_0$ and $\omega$ is in $\sT_i$, the bar recursor $\BR_{\omega,g,h}$ can be defined in $\sT_{i+3}$ (see \cite[Corollary 3.5]{OliStei(2017.0)}). But then $\Phi_{\omega}$ is also  definable in $\sT_{i+3}$, and since a derivational complexity function for $\succ$ is given by $\lambda x.|\Phi_{\omega}(x)|$, this gives us (\ref{barclosa}). Part (\ref{barclosb}) follows analogously using Howard's more refined result for lower types \cite{Howard(1981.0)}.
\end{proof}
Corollary \ref{barclosure} is by no means exhaustive. For example, Howard's closure theorem \cite{Howard(1981.0)} is extended to fragments of the Grzegorzyk hierarchy by Kreuzer \cite{Kreuzer(2012.1)}, though it is unclear whether this would be applicable here, since these fragments do not have access to the full recursor of lowest type. Note that it could also be that a more carefully analysis of the particular form of bar recursion we use could lead to a significantly improved version of Corollary \ref{barclosure}. We conjecture that our bar recursive program is actually closed on fragments of System $\sT$, although we leave this open for now.

However, all of this demonstrates how our approach of extracting concrete programs and then appealing to computability theory of those programs leads to extremely general complexity results. In the next section, we show that the situation improves further if we strengthen our hypothesis by replacing the modulus of minimal wellfoundedness by some explicit recursor.

%%%%%%%%%%%%%%%%%%%%%%%%%%%%%%%%%%%%%%%%%%%%%%%%%
\subsection{Derivation functions for inductively wellfounded orders}
%%%%%%%%%%%%%%%%%%%%%%%%%%%%%%%%%%%%%%%%%%%%%%%%%
\label{sec-simp-prclosure}

We now demonstrate how Theorem \ref{thm-abs} and the associated complexity bounds can be improved if we take as a stronger premise that minimal sequences are inductively wellfounded with respect to some concrete relation $\rrhd$ on $X\times\arr{X}$.
\begin{lemma}
\label{lem-interdef}
Suppose that $\rrhd$ is some inductively wellfounded relation on $X\times\arr{X}$, and that Suppose that $c_\succ,c_\rhd:X\to X^\ast$ are some fixed terms of $T_0$ which form branching moduli for $\succ$ and $\rhd$. Then there is a functional $\Gamma:X\times\arr{X}\to X^\ast$ definable in $T_0+\rec^{\rhd,X^\ast}+\rec^{\rrhd,X^\ast}$ which satisfies
\begin{equation*}
\Gamma(x,u)=x\ast\cnc_{y\in c_\succ(x)} \tilde R_{x,u}(y) 
\end{equation*}
where $\tilde R_{x,u}:X\to X^\ast$ satisfies
\begin{equation*}
R_{x,u}(y):=\begin{cases}u_i[y] & \mbox{if $y\preceq c_\rhd(x)_i$ for some $i<|c_\rhd(x)|$}\\ \Gamma(y,v) & \mbox{if $(x,u)\rrhd (y,v)$}\\ [] & \mbox{otherwise}\end{cases}
\end{equation*}
for $v:=\pnc_{z\in c_\rhd(y)}R_{x,u}(z)$. Suppose that $\rrhd$ satisfies
\begin{equation}
\label{eqn-interdefhyp}
T_\rhd(x,u)\wedge T_\rhd(y,v)\wedge x\gg y\to (x,u)\rrhd (y,v)
\end{equation}
for all $(x,u),(y,v)$. Then for any $x,u\in T_\rhd$ we have
\begin{equation}
\label{eqn-interdef}
(\forall a,b)(\pair{a,b}\in T_\rhd\wedge C_\gg(a\ast x)\to\Psi_\omega(a\ast x,b\ast u)=\Gamma(x,u))
\end{equation}
whenever $\omega$ is a modulus of minimal wellfoundedness, where $\Psi_\omega$ is defined as in Lemma \ref{lem-absdef}.
\end{lemma}

\begin{proof}
That $\Gamma$ is definable in $T_0+\rec^{\rhd,X^\ast}+\rec^{\rrhd,X^\ast}$ is just a simple adaptation of the proof of Lemma \ref{lem-absdef}. We prove (\ref{eqn-interdef}) by induction on $\rrhd$. Note that $C_\gg(a\ast x)$ implies that $\omega(\widehat{a\ast x,b\ast u})<|a|+1$ and thus $\Psi_\omega(a\ast x,b\ast u)=x\ast\cnc_{y\in c_\succ(x)} R_{a\ast x,b\ast u}(y)$. So we're done if we can show that $R_{a\ast x,b\ast u}(y)=\tilde R_{x,u}(y)$ for all $y\prec x$.

We do this by a side induction on $\rhd$, so fix some $y$ and assume that $R_{a\ast x,b\ast u}(z)=\tilde R_{x,u}(z)$ for all $z\lhd y$. We only need to check the case $x\gg y$, where we must show that $(x,u)\rrhd (y,v)$ for $v:=\pnc_{z\in c_\rhd(y)} \tilde R_{x,u}(z)=\pnc_{z\in c_\rhd(y)} R_{a\ast x,b\ast u}(z)$. But since $T_\rhd(y,v)$ this follows by (\ref{eqn-interdefhyp}) and thus
\begin{equation*}
\tilde R_{x,u}(y)=\Gamma(y,v)=\Psi_\omega(a\ast x\ast y,b\ast u\ast v)=R_{a\ast x,b\ast u}(y)
\end{equation*}
by the main induction hypothesis.\end{proof}

\begin{corollary}
\label{cor-prcomp}
Under the conditions of Lemma \ref{lem-interdef}, the functional $\Phi_\omega$ in Theorem \ref{thm-abs} is definable in $\sT_0+\rec^{\rhd,X^\ast}+\rec^{\rrhd,X^\ast}$ for any modulus of minimal wellfoundedness $\omega$.
\end{corollary}

\begin{proof}
This follows from (\ref{eqn-interdef}), setting $|a|=0$ and using induction over $\rhd$.
\end{proof}

We can now give a more direct formulation of Corollary \ref{barclosure}:

\begin{corollary}
\label{prclosure}
Suppose that $\succ$ is a binary relation whose branching modulus for both $\succ$ and $\rhd$ is definable in $\sT_0$, and that $\rec^{\rhd,X^\ast}$ is also definable in $\sT_0$. Under the conditions of Lemma \ref{lem-interdef}, whenever $\rec^{\rrhd,X^\ast}$ is definable in $\sT_i$, the derivational complexity $\succ$ is bounded by some function also in $\sT_i$. In particular:
\begin{enumerate}[(a)]

\item\label{prclosa} For $i=0$, the derivational complexity is bounded by a primitive recursive function;

\item\label{barclosb} For $i=1$, the derivational complexity is bounded by a multiple recursive function.
\end{enumerate}

\end{corollary}

%%%%%%%%%%%%%%%%%%%%%%%%%%%%%%%%%%%%%%%%%%%%%%%%%
%%%%%%%%%%%%%%%%%%%%%%%%%%%%%%%%%%%%%%%%%%%%%%%%%
\section{Application: Path orders and term rewriting}
%%%%%%%%%%%%%%%%%%%%%%%%%%%%%%%%%%%%%%%%%%%%%%%%%
%%%%%%%%%%%%%%%%%%%%%%%%%%%%%%%%%%%%%%%%%%%%%%%%%
\label{sec-trs}

We conclude by sketching how our abstract results can be applied in the special case where $X$ denotes a set of terms in some programming language, which we take here to be a simple term rewrite system. In this we show how the formalization of Buchholtz \cite{Buchholz(1995.0)} can be incorporated into our setting. The difference here is that we directly construct derivation functions in fragments of System $\sT$, rather than formalizing wellfoundedness proofs in fragments of Peano arithmetic.

%%%%%%%%%%%%%%%%%%%%%%%%%%%%%%%%%%%%%%%%%%%%%%%%%
\subsection{$(X,\rhd)$ as a term structure}
%%%%%%%%%%%%%%%%%%%%%%%%%%%%%%%%%%%%%%%%%%%%%%%%%
\label{sec-trs-terms}

Let $X$ now be instantiated as the set of terms ranging over some countable set of variables and some finite signature $\{f_1,\ldots,f_k\}$, where we assume for simplicity that each $f_i$ has a fixed arity (note that this latter restriction is not essential: see \cite[Section 3]{Buchholz(1995.0)}). Clearly $X$ can be arithmetized, and as in \cite{Buchholz(1995.0)}, we can assign each term a size as follows:
\begin{enumerate}[(i)]

\item $|x_i|:=i$;

\item $|f_j(t_1,\ldots,t_n)|:=\max\{n,|t_1|,\ldots,|t_n|\}+1$.

\end{enumerate}
and assume w.l.o.g. exists some monotone function $h$ such that $|t|\leq t<h(|t|)$ for all $t$. Let $\rhd$ denote the immediate subterm relation: in other words, $f(t_1,\ldots,t_n)\rhd t_i$ for all $i=1,\ldots,n$. Clearly, $s\lhd t$ implies $|s|<|t|$, and so the recursion over $\rhd$ is definable from the usual G\"{o}del recursor over $>$. In particular, $\rec^{\rhd,X^\ast}$ is definable in $\sT_0$.

%%%%%%%%%%%%%%%%%%%%%%%%%%%%%%%%%%%%%%%%%%%%%%%%%
\subsection{Approximations to recursive path orders}
%%%%%%%%%%%%%%%%%%%%%%%%%%%%%%%%%%%%%%%%%%%%%%%%%
\label{sec-trs-po}

In general, recursive path order on terms can be characterized in the abstract as follows: We set $t=f(t_1,\ldots,t_n)\succ s$ if either
\begin{enumerate}[(a)]

\item $t_i\succeq s$ for some $i=1,\ldots,n$; 

\item $t\succ_0 s$ and $t\succ s_i$ for all subterms $s_i$ of $s$,

\end{enumerate}
where typically $\succ_0$ is recursively defined in terms of $\succ$ itself. Note that by (a), $\succ$ contains the subterm relation, which means that it is a \emph{simplification order}. Condition (b) is closely related to the notion of a \emph{lifting} as studies in e.g. \cite{FerZan(1995.0)}. In any case, such an order is clearly a decomposition in the sense of our Definition \ref{defn-decomp} relative to $\gg$, where $t\gg s$ denotes the second case above. 

Recursive path orders of this kind are fundamental tools in the theory of term rewriting, as they provide us with a criterion for checking if finitely defined term rewrite system $\R$ is terminating. Here we would work with orders $\succ$ which are closed under contexts and substitutions, and then whenever the rules $l\to r$ of $\R$ satisfy $l\succ r$, then $\R$ guaranteed to be terminating. The main challenge is always to show that $\succ$ itself is wellfounded.

When it comes to computing complexity bounds, the first issue is that in general, recursive path orders are not finitely branching, and as a result proofs of wellfoundedness tend to use rather heavy proof theoretic machinery such as Kruskal's theorem. However, this is overcome in \cite{Buchholz(1995.0)} by considering finitary variants of the usual path orders, whose wellfoundedness can be proven in low fragments of arithmetic.

One can describe Buchholz' idea in a slightly more general form as follows: for some primitive recursive function $b:\NN\to\NN$ define the bounded $b$-approximation $\succ_b$ of $\succ$ by $t=f(t_1,\ldots,t_n)\succ_b s$ if $b(|t|)\geq |s|$ and either of the following hold:
\begin{enumerate}[(a)]

\item $t_i\succeq_b s$ for some $i=1,\ldots,n$; 

\item $t\gg_b s$ and $t\succ_b s_i$ for all subterms $s_i$ of $s$,

\end{enumerate}
where now $\gg_b$ is recursively defined in terms of $\succ_b$. Not only is $\succ_b$ now by definition finitely branching, but assuming that $\succ_b$ is a primitive recursive relation, as it invariably is, then $\succ_b$ is computably finitely branching: Because $t\succ_b s$ only if $b(|t|)\geq |s|$ and hence $h(b(|t|))\geq h(|s|)>s$, and we can therefore take the branching function $c_\succ(t)\in X^\ast$ to be the primitive recursively definable sequence consisting of exactly those terms $s\leq h(b(|t|))$ satisfying $s\prec_b t$.

The crux of the idea is the following: Suppose that for any $\R$ reducing under $\succ$, there is some $b$ such that $\R$ is reducing under $\succ_b$, in other words, for any \emph{fixed} $\R$ we can find a finitely branching approximation $\succ_b$ of $\succ$ sufficient for proving wellfoundedness of $\R$. Then the derivational complexity of $\R$ is bounded by the derivational complexity of $\succ_b$. 

Thus our complexity results in Corollaries \ref{barclosure} and \ref{prclosure} provide us with a means of bounding the derivational complexity of rewrite systems, and the generality of our results suggest that they are applicable to a wide range of different path orders. We finish by sketching a simple example.

%%%%%%%%%%%%%%%%%%%%%%%%%%%%%%%%%%%%%%%%%%%%%%%%%
\subsection{Example: the multiset path order}
%%%%%%%%%%%%%%%%%%%%%%%%%%%%%%%%%%%%%%%%%%%%%%%%%
\label{sec-trs-mult}

We now show how the well known primitive recursive bound on the complexity of rewrite systems terminating under the multiset path order can be reobtained in our setting. A simple variant of the path order is obtained by instantiating $\gg$ as $f(t_1,\ldots,t_n)\gg s$ if
\begin{itemize}

\item $s=g(s_1,\ldots,s_m)$ and $f>_F g$, or

\item $s=f(s_1,\ldots,s_n)$ and $t_i\succ s_i$ for some $i=1,\ldots,n$ and $s_j=t_j$ for all $j\neq i$.

\end{itemize}
where $>_F$ is some wellfounded relation on function symbols. It turns out that any rewrite system reducing under the multiset path order is also reducing under the approximate order $\succ_k$ in which the bounding function $b$ is simply $b(n)=n+k$, and $k$ is some sufficiently large number which can be effectively computed from the rules of the rewrite system.

Given in full, then, the approximate multiset path order $\succ_k$ is defined as follows: $t=f(t_1,\ldots,t_n)\succ_k s$ iff $k+|t|\geq |s|$ and either
\begin{enumerate}[(a)]

\item $t_i\succeq_k s$ for some $i=1,\ldots,n$; 

\item $f(t_1,\ldots,t_n)\gg_k s$

\end{enumerate}

where $f(t_1,\ldots,t_n)\gg_k s$ iff

\begin{enumerate}[(i)]

\item $s=g(s_1,\ldots,s_m)$ with $f>_F g$ and $t\succ_k s_i$ for all $i$;

\item $s=f(s_1,\ldots,s_n)$ and $t\succ_k s_i$ for all $i$ and $t_i\succ_k s_i$ for some $i$ and $s_j=t_j$ for all $j\neq i$.

\end{enumerate}

Now, define the relation $\rrhd_k$ on $X\times\arr{X}$ as follows: $\pair{f(t_1,\ldots,t_n),u}\rrhd_k \pair{g(s_1,\ldots,s_m),v}$ iff
\begin{equation*}
f>_F g \mbox{ \ \ \ or \ \ \ }f=g\wedge (\exists i)(u_i\supset v_i\wedge (\forall j\neq i)(u_j\supseteq v_j)).
\end{equation*}
It is easy to see that if $T_\rhd(t,u)\wedge T_\rhd(s,v)\wedge t\gg_k s$ then $\pair{t,u}\rrhd_{k}\pair{s,v}$: In the case that (i) holds then $f>_T g$ so this is clearly true, while if (ii) holds there is some $i$ such that $t_i\succ_k s_i$ but $t_j=s_j$ otherwise. But $T_\rhd(t,u)$ implies that $T(t_i,u_i)$, and analogously $T_\rhd(s,v)$ implies $T(s_i,v_i)$, and so $t_i\succ_k s_i$ implies that $v_i$ is a subsequence of $u_i$. Similarly, we must have $u_j=v_j$ otherwise.

Not only is $\rrhd_k$ clearly primitive recursive, but it is not difficult to show that $\rec^{\rrhd_k,X^\ast}$ is definable in $\sT_0$: This is just a bounded recursion in the first component, while in the second component we can find an encoding of $\arr{T}$ into $\NN$ such that $(\exists i)(u_i\supset  v_i \wedge(\forall j\neq i)( u_j\supseteq v_j))$ implies that $u>v$.

Therefore by Corollary \ref{prclosure}, the derivational complexity of $\succ_k$ is bounded by a primitive recursive function, and therefore the same is true for any rewrite system $\R$ reducing under $\succ$.

Both the multiset and lexicographic path orders are studied in more detail by the author together with Georg Moser in \cite{MosPow(2015.0)}, where a more detailed construction of the derivational complexity functions is given than our brief sketch here, although the main results of this paper constitute a considerable generalisation of \cite{MosPow(2015.0)}.

%%%%%%%%%%%%%%%%%%%%%%%%%%%%%%%%%%%%%%%%%%%%%%%%%
%%%%%%%%%%%%%%%%%%%%%%%%%%%%%%%%%%%%%%%%%%%%%%%%%
\section{Conclusion}
%%%%%%%%%%%%%%%%%%%%%%%%%%%%%%%%%%%%%%%%%%%%%%%%%
%%%%%%%%%%%%%%%%%%%%%%%%%%%%%%%%%%%%%%%%%%%%%%%%%
\label{sec-conc}

The main result of this paper was a constructive analysis of the wellfoundedness of abstract path orders, which in particular subsumes the usual recursive path orders encountered in the term rewriting literature. A such, the paper is a contribution to the proof theoretic analysis of termination, which has seen a number of recent developments \cite{BerOliSte(2015.0),BerSte(2015.0),FriYokSte(2017.0),MosPow(2015.0)}. As a side product of our theoretical work, we provide a series of metatheorems which allow us to relate the complexity of a wellfounded order to some subrecurive system of functionals. While we only sketched an illustration of this in Section \ref{sec-trs}, we believe that the formal extraction of programs from termination proofs has a great deal of potential in providing upper bounds on the complexity of programs, and in this article hope to have provided a promising first step in this direction. We conclude with a collection of open problems.

An obvious direction for future research is to use the techniques presented here to obtain new bounds and metatheorems for the complexity of concrete termination orders. While we mentioned the well-known recursive path orders as a simple example of where Corollary \ref{prclosure} could be applied, of particular interest would be the analysis of termination orders for which an upper bound on the induced derivational complexity is not known. 

Most termination orders in the literature work on sets of first order terms. However, up to the very final section we do not assume anything about the structure of $X$, and it would be interesting to find our whether our termination arguments can be applied to more interesting structures. In particular, Goubault-Larrecq \cite{GLar(2001.0)} considers wellfounded orders on graphs, automata and higher-order functionals, and it would be intriguing to see whether any of these are subsumed by our abstract principle, and whether any meaningful complexity results could be obtained.

In our approach, we establish complexity bounds by extracting higher-order recursive programs in some subrecursive calculus of functionals, and looking at the type $1$ functions definable in these calculi. A number of similar proof theoretic investigations of path orders and abstract notions of termination exist in the literature, notably those due to Weiermann \cite{Weiermann(1994.0),Weiermann(1998.0)} which are based on an intricate ordinal analysis. It would be instructive to make more precise how our framework based on variants of bar recursion compares to his.

Finally, as briefly mentioned in Section \ref{sec-simp-barclosure}, it would be interesting to formally establish a set of closure properties along the lines of \cite{EOP(2011.0),OliStei(2017.0),Schwichtenberg(1979.0)} for \emph{finitely branching} bar recursion, which would give a direct correspondence between the subrecursive strength of bar recursors and the derivational complexity of abstract orders.

\noindent\textbf{Acknowledgements.} I am indebted to Georg Moser for suggesting a proof theoretic study of termination principles, and in particular for pointing out that the results of \cite{Buchholz(1995.0)} can be viewed in a more abstract way. This work was partially supported by the Austrian Science Fund (FWF) project P 25781-N15.

\appendix

%%%%%%%%%%%%%%%%%%%%%%%%%%%%%%%%%%%%%%%%%%%%%%%%%
%%%%%%%%%%%%%%%%%%%%%%%%%%%%%%%%%%%%%%%%%%%%%%%%%
\section{Appendix}
%%%%%%%%%%%%%%%%%%%%%%%%%%%%%%%%%%%%%%%%%%%%%%%%%
%%%%%%%%%%%%%%%%%%%%%%%%%%%%%%%%%%%%%%%%%%%%%%%%%
\label{sec-app}

\begin{proof}
[Proof of Lemma \ref{lem-absdef}]
We define functions $g:Y^\ast\to X^\ast$ and $h:Y^\ast\to (Y\to X^\ast)\to X^\ast$ by
\begin{equation*}
\begin{aligned}
g(a,b)&:=[]\\
h(a,b)(p)&:=\begin{cases}[] & \mbox{if $|a|=0$}\\ \bar{a}\ast\cnc_{y\in c_\succ(\bar{a})}\rec^\rhd_{f_{\bar{a},\bar{b},p}}(y) & \mbox{otherwise}\end{cases}
\end{aligned}
\end{equation*}
where $f:X\to \arr{X}\to (Y\to X^\ast)\to X\to (X\to X^\ast)\to X^\ast$ is defined by
\begin{equation*}
f_{x,u,p}(y)(q):=\begin{cases}u_i[y] & \mbox{if $y\preceq c_\rhd(x)_i$ for some $i<|c_\rhd(x)|$}\\ p(y,\pnc_{z\in c_\rhd(y)} q(z)) & \mbox{otherwise}\end{cases}
\end{equation*}
and $d[y]$ is defined as in the statement of the lemma. Now, it is not difficult to see that since $\succ$ is primitive recursive and $c_\succ$ a branching modulus then $T(y,d)$ is primitive recursive, and thus so is $d[y]$ since this can be computed via a bounded search. Moreover, the case distinction in the definition of $f$ is primitive recursively decidable, and so the functional as a whole is clearly definable in $\sT_0$. It is obvious then that $f$ is definable in $\sT_0+\rec^{\rhd,X^\ast}$ and so 
\begin{equation*}
\Psi_\omega:=\BR^{Y,X^\ast}_{\omega,g,h}
\end{equation*}
is definable in $\sT_0+\rec^{\rhd,X^\ast}+\BR^{Y,X^\ast}$. To see that it satisfies the relevant equations is just a matter of unwinding definitions: We have $\Psi_\omega(a,b)=h(a,b)(\ldots)=[]$ is $|a|=0$ and $\Psi_\omega(a,b)=g(a,b)=[]$ if $\omega(\ext{a,b})<|a|$, and otherwise
\begin{equation*}
\begin{aligned}
\Psi_\omega(a,b)&=\bar{a}\ast\cnc_{y\in c_\succ(\bar{a})}R_{\omega,a,b}(y)
\end{aligned}
\end{equation*}
for $R_{\omega,a,b}:=\rec^\rhd_{f_{\bar{a},\bar{b},p}}$ and $p:=\lambda x,u.\Psi_\omega(a\ast x,b\ast u)$. But then 
\begin{equation*}
\begin{aligned}
R_{\omega,a,b}(y)&=f_{\bar{a},\bar{b},p}(y)(\lambda z\lhd y\; . \; R_{\omega,a,b}(z))\\
&=\begin{cases}\bar{b}_i[y] & \mbox{if $y\preceq c_\rhd(\bar{a})_i$ for some $i<|c_\rhd(x)|$}\\
p(y,\pnc_{z\in c_\rhd(y)} R_{\omega,a,b}(z)) & \mbox{otherwise}\end{cases}
\end{aligned}
\end{equation*}
and in the second line
\begin{equation*}
p(y,\pnc_{z\in c_\rhd(y)} R_{\omega,a,b}(z))=\Psi_\omega(a\ast y,b\ast \pnc_{z\in c_\rhd(y)} R_{\omega,a,b}(z))
\end{equation*}
which completes the proof.
\end{proof}

\bibliographystyle{plain}
\bibliography{/home/thomas/Dropbox/tp}
\end{document}